\def\verdate{10/06/2016}\def\ver{0.46}

\documentclass[11pt,a4paper]{article}

\usepackage{fullpage}
\usepackage{graphicx}
\usepackage{subfigure}
\usepackage{amsmath,amssymb}
\usepackage{dsfont}
\usepackage{hyperref}
\usepackage{cite}
\usepackage{color}
\usepackage{tikz}

\newtheorem{definition}{Definition}
\newtheorem{theorem}{Theorem}
\newtheorem{lemma}{Lemma} 
\newcommand{\eg}{\emph{e.g.\/}}
\newcommand{\ie}{\emph{i.e.\/}}
\newcommand{\etal}{\emph{et al.}}
\DeclareMathOperator{\tr}{tr}
\newcommand{\Tr}{\tr}
\newcommand{\C}{\ensuremath{\mathds{C}}}
\newcommand{\Cplx}{\ensuremath{\C}}
\newcommand{\bbC}{\ensuremath{\Cplx}}
\newcommand{\ket}[1]{\ensuremath{|#1\rangle}}
\newcommand{\bra}[1]{\ensuremath{\langle#1|}}
\newcommand{\ketbra}[2]{\ensuremath{\ket{#1}\bra{#2}}}
\newcommand{\braket}[2]{\ensuremath{\langle{#1}|{#2}\rangle}}
\newcommand{\1}{{\rm 1\hspace{-0.9mm}l}}
\newcommand{\Id}{\1}
\newcommand{\ii}{\mathrm{i}}
\DeclareMathOperator{\diag}{diag}

\newenvironment{proof}[1][Proof]{\noindent\textbf{#1.} }{\hfill 
\rule{0.5em}{0.5em}}

\title{Lively quantum walks on cycles}

\author{Przemys{\l}aw Sadowski${}^1$ \quad Jaros{\l}aw Adam Miszczak${}^{1,2}$ \quad Mateusz Ostaszewski${}^1$\footnote{Corresponding author: mostaszewski@iitis.pl} \\[12pt]
  ${}^1$Institute of Theoretical and Applied Informatics, Polish Academy of Sciences,\\
  Ba{\l}tycka 5, 44-100 Gliwice, Poland\\[6pt]
  ${}^2$Applied Logic, Philosophy and History of Science group,\\ 
University of Cagliari,\\ Via~Is Mirrionis 1, 09123 Cagliari, Italy
}

\date{\verdate\ (v. \ver)}

\begin{document}

\maketitle

\begin{abstract}

We introduce a family of quantum walks on cycles parametrized by their 
liveliness, defined by the ability to execute a long-range move. We investigate 
the behaviour of the probability distribution and time-averaged probability 
distribution. We show that the liveliness parameter, controlling the magnitude 
of the additional long-range move, has a direct impact on the periodicity of 
the limiting distribution. We also show that the introduced model provides a 
method for network exploration which is robust against trapping.


\end{abstract}

\section{Introduction}
Quantum walks \cite{reitzner12walks,vanegas12review,portugal-book}, quantum
counterparts of classical Markov processes, provide a powerful method for
developing new quantum algorithms~\cite{ambainis03algorithmic} and
protocols~\cite{difranco11mimicking,mcgettrick14cycles,miszczak14quantum,sadowski14efficient,pawela15generalized}.
As quantum protocols have to be executed on pair with classical protocols
controlling distant parts of a quantum network, quantum walks have to include
elements enabling them to adapt to the current structure of the network. The
methods of adapting classical algorithms for the purpose of quantum networks are
currently under an active investigation~\cite{vanmeter13path} and include the
application of game theory in a complex quantum network with interacting parties
\cite{pawela13cooperative}.

Quantum walks on cycles can be used as a simple and very powerful model for the
purpose of modeling quantum and hybrid classical-quantum networks. In particular,
in~\cite{miszczak14quantum}, the authors have developed a model that can be used to
analyze the scenario of exploring quantum networks with a distracted sense of
direction. By using this model, it is possible to study the
behavior of quantum mobile agents operating with non-adaptive and adaptive
strategies that can be employed in this scenario.

The presented work introduces a family of quantum walks on cycles
with \emph{liveliness},
corresponding to the ability to execute a long-range
move. The introduced family is parametrized by the liveliness
parameter, which is used to control the magnitude of the additional long-range
move. In particular, the proposed family contains lazy quantum walks, which can
be introduced as quantum walks with liveliness equal to~0. We investigate the
behavior of the probability distribution and time-averaged probability
distribution~\cite{aharonov01walks} for the introduced family and generalize the
results obtained by Bednarska \etal~\cite{bednarska03walks}. We show that the
liveliness parameter has a direct impact on the periodicity of the limiting
distribution. We also show that the
introduced model provides a method for network exploration which is robust against trapping.

This paper is organized as follows.
In Section~\ref{sec:model} we introduce the model of lively quantum walks on
cycles.
In Section~\ref{sec:periodicity} we study the behavior of the time-averaged
limiting distribution of the introduced model and discuss its periodicity.
In Section~\ref{sec:liveliness} we prove that the introduced model allows the
improvement of the quantum network exploration. This is achieved by demonstrating that our model can be
used to avoid trapping and to counteract 
malfunctions in the network.
Finally, in Section~\ref{sec:concluding}, we discuss the possible applications
and extensions of the introduced model.

\section{Lively quantum walks}\label{sec:model}
Let us first consider a cycle with $n$ nodes and define a standard model of
quantum walk. The position of a walker during a quantum walk executed on such
cycle is described by a vector in $n$-dimensional complex space $\Cplx^n$. The
state space is of the form $\C^2\otimes\C^n$. Quantum walk process is defined in
the situation by the shift operator 
\begin{equation}
\begin{split}
S(\ket{0}\otimes\ket{i}) &= \ket{0}\otimes\ket{ i+1 \pmod n},\\
S(\ket{1}\otimes\ket{i}) &= \ket{1}\otimes\ket{i-1 \pmod n},
\end{split}
\end{equation}
or equivalently
\begin{equation}
S = \ketbra{0}{0}\otimes\sum\limits_{i=0}^{n-1}\ketbra{i+1}{i}+
\ketbra{1}{1}\otimes\sum\limits_{i=0}^{n-1}\ketbra{i-1}{i},
\end{equation}
where addition is modulo $n$.
In this case the walker has for her disposal two directions -- $\ket{0}$ 
(right) 
and $\ket{1}$ (left) -- represented in the two dimensional Hilbert space 
$\C^2$, 
corresponding to the coin used in the classical random walk. 
The evolution operator is defined as
\begin{equation}
U = S(C\otimes\Id_n),
\end{equation}
where $C$ is the coin operator which acts on coin space $\C^2$.

Let us now assume that the coin register used to control a quantum
walker is represented by a vector in $\Cplx^3$ (\ie\ by a qutrit) and thus, during
each step, the walker can change its position according to one of three possible
states of the coin register. Using this setup we define \emph{lively quantum
walk} on cycles as follows.

\begin{definition}[Lively quantum walk on a cycle]
Lively quantum walk on a
$n$-dimensional cycle with \emph{liveliness} $0\leq a\leq n-1$, is defined by 
the shift operator $S^{(n,a)}\in L(\Cplx^3\otimes\Cplx^n)$ of the form
\begin{equation}
S^{(n,a)} = \sum_{x=0}^{n-1} S_x^{(n,a)},\label{eq:shift}
\end{equation}
where
\begin{equation}
\begin{split}
S_x^{(n,a)} &=  \ketbra{0}{0}\otimes\ketbra{x-1 \pmod n}{x}\\
        &+ \ketbra{1}{1}\otimes\ketbra{x+1 \pmod n}{x} \\
        &+ \ketbra{2}{2}\otimes\ketbra{x+a \pmod n}{x}.
\end{split}\label{eq:shift-at-position}
\end{equation}
\end{definition}

For the case $a=0$ the above definition reduces to \emph{lazy quantum walk}
(\ie\ a quantum walk with no liveliness). 
One should also note that if the introduced model would included jumps with $-a$
parameter we could restrict our considerations to the case where $a\le\lfloor
\frac{n}{2} \rfloor$.

For the small number of nodes, the existence of the additional
connections can be used to model the transition from a cycle to the full
network. For example, for $n=6$, the lively quantum walk with $a=2$ is equivalent to
the quantum walk on the total network. However, in order to study the parametrized family of processes on graphs with different degree of connectivity, additional connections, and thus larger coins, are needed.

Operator $S_x^{(n,a)}$ acts on position $x$ by shifting it by $+1$, $-1$ or
by the value specified by the liveliness parameter $a\in\{0,1,\dots,n-1\}$. The case
$a=0$ is identical to the case $a=n$.


The coin operator used in the further considerations is defined by the
Grover operator
\begin{equation}
G = 2 \ketbra{c_1}{c_1}-\Id_3=\left(
\begin{smallmatrix}
	-\frac{1}{3} & \frac{2}{3} & \frac{2}{3} \\
	\frac{2}{3} & -\frac{1}{3} & \frac{2}{3} \\
	\frac{2}{3} & \frac{2}{3} & -\frac{1}{3} \\
\end{smallmatrix}\right),
\end{equation}
where 
\begin{equation}\label{eqn:init0}
	\ket{c_1}=\frac{1}{\sqrt{3}}\left(\ket{0}+\ket{1}+\ket{2}\right).
\end{equation}

Using the above we define the walk operator for the lively walk on cycle as
\begin{equation}
U=\left(\sum_{x=0}^{n-1}S_x^{(n,a)}\right) (G\otimes \Id_n).
\end{equation}

\section{Limiting distribution periodicity}\label{sec:periodicity}
We start with the proof of the periodicity of the limiting distribution for the
introduced model. Let us introduce the time-averaged probability distribution
for the quantum walk as follows.

\begin{definition}
We define the time-averaged probability distribution at position $x$ for a 
unitary process $U$ as
\begin{equation}
\Pi(x) = \lim_{N\rightarrow \infty} \frac{1}{N} \sum_{t=0}^N P(x, t),
\end{equation}
where $P(x,t)$ denotes the probability of measuring position $x$ after $t$ steps
\begin{equation}
P(x,t) =\sum_{c=0}^2 |\bra{c,x}U^t\ket{\psi_0}|^2
\end{equation}
and $\ket{\psi_0}$ is arbitrary initial state.
\label{def::limiting}
\end{definition}

Let us now consider a lively quantum walk with $n$ nodes and the step size $a$
chosen in such a way that there is a common divisor of both numbers.

\begin{theorem}\label{th:cyclic}
If $\mathrm{GCD}(a, n)>1$ then the limiting time-averaged probability
distribution is periodic with period equal to $\mathrm{GCD}(a, n)$.
\end{theorem}

First, let us note that the spectrum of the walk operator is conveniently 
expressed using Fourier basis at the position register.

\begin{lemma}\label{lem:eigenvalues}
The walk operator $U = S^{(n,a)}(G\otimes\1_n)\in L(\bbC^3\otimes \bbC^n 
)$ 
has 
eigenvalues 
$\lambda_{k,j}$ with 
corresponding eigenvectors 
$\ket{\psi_{k,j}}=\ket{v_{k,j}}\otimes\ket{\phi_k}\in\bbC^3\otimes\bbC^n$ 
satisfying the 
equation
\begin{equation}\label{eq:g-diag-eigenstates}
\diag(e^{-\ii k},e^{\ii k},e^{-\ii ka})G\ket{v_{k,j}} = 
\lambda_{k,j}\ket{v_{k,j}},
\end{equation}
where $\ket{\phi_k}=\sum e^{-\ii kx}\ket{x}$ for $k = \frac{2\pi l}{n}, 
l=0,\ldots,n-1$, $j\in\{0,1,2\}$.
\end{lemma}

\def\zerodirection{0}
\def\oneodirection{1}
\def\twodirection{2}
\begin{proof}
We analyse the action of the step on $\ket{d}\otimes\ket{\phi_k}$ for basis 
state $d\in\{\zerodirection, \oneodirection, \twodirection \}$ and note that 
the step 
operator acting on states with Fourier states on position register results in 
a relative phase
$S^{(n,a)}\ket{\zerodirection}\otimes\ket{\phi_k}=e^{\ii 
k}\ket{\zerodirection}\otimes\ket{\phi_k}$,
$S^{(n,a)}\ket{\oneodirection}\otimes\ket{\phi_k}=e^{-\ii 
k}\ket{\oneodirection}\otimes\ket{\phi_k}$
and
$S^{(n,a)}\ket{\twodirection}\otimes\ket{\phi_k}=e^{-\ii 
ka}\ket{\twodirection}\otimes\ket{\phi_k}$.
Thus we can reduce the dynamics of the states of the form 
$\ket{v}\otimes\ket{\phi_k}$ so that we consider the $\ket{v}$ subsystem only 
and we substitute the step operator with $\diag(e^{-\ii k},e^{\ii k},e^{-\ii 
ka})$.
Therefore any eigenvector $\ket{v_{k,j}}$ of the form given in Eq. 
(\ref{eq:g-diag-eigenstates})
corresponds to an eigenvector of the form $\ket{v_{k,j}}\otimes\ket{\phi_k}$ of 
the walk operator $U$.
\end{proof}

In the context of the limiting distribution we emphasise the fact that 
for any eigenvector of the form 
$\ket{\psi_{k,j}}=\ket{v_{k,j}}\otimes\ket{\phi_k}$
the probability distribution at the coin register
\begin{equation}
|(\bra{d}\otimes\bra{x})\ket{\psi_{k,j}}|^2 = 
|\braket{d}{v_{k,j}}|^2\label{eqn:transition-invariant}
\end{equation}
is position independent. This property can be applied into limiting 
distribution formula
\begin{equation}
\Pi(x)=\sum_{\lambda}\sum_d\sum_{(k,j),(k',j')\in V_\lambda 
}a_{k,j}^{\phantom*} 
a_{k',j'}^* \braket{d,x}{\psi_{k,j}}\braket{\psi_{k',j'}}{d,x},
\end{equation}
where $x$ is the position, $\ket{\psi_0}=\sum a_{k,l} \ket{\psi_{k,l}}$, 
$V_\lambda$ are
indices of $\lambda$-eigenvectors such that 
$V_{\lambda}=\{(k,j):\lambda_{k,j}=\lambda\}$.
It is straightforward from Eq.~(\ref{eqn:transition-invariant}) that for 
1-dimensional eigenspaces the probability is 
transition invariant
\begin{equation}
a_{k,j}^{\phantom\dagger} a_{k,j}^* 
\braket{d,x}{\psi_{k,j}}\braket{\psi_{k,j}}{d,x} =
a_{k,j}^{\phantom\dagger} a_{k,j}^* |\braket{d}{v_{k,j}}|^2
=
a_{k,j}^{\phantom\dagger} a_{k,j}^* 
\braket{d,x'}{\psi_{k,j}}\braket{\psi_{k,j}}{d,x'},
\end{equation}
for $x,x'=1,\dots,n$. For higher-dimensional eigenspaces we are concerned with
relative phase during transition. In other words, one is assured that for two
positions $x, x'$ the modules of the terms
$\braket{d,x}{\psi_i}\braket{\psi_j}{d,x}$ and
$\braket{d,x'}{\psi_i}\braket{\psi_j}{d,x'}$ are equal, however they may differ
in phase. Here we prove that the dimensionality of eigenspaces of eigenvalues
$\lambda_{k,j}$ is higher than one if the relation $k = \frac{2\pi}{n}l$ is
satisfied for $l$ being the multiplication of $\frac{n}{\mathrm{GCD}(n,a)}$. We do
not prove that for $k \ne \frac{2\pi l}{\mathrm{GCD}(n,a)}$ the
$\lambda_{k,j}$-eigenspace is one-dimensional, but the influence of the cases
when it is not true is negligible.

\begin{lemma}\label{lem:multi-eigenvalues}
For $k=\frac{2\pi l}{n}$ and $\frac{n}{\mathrm{GCD}(n,a)}|l$ we have that 
$\lambda_{k,0}=1$ is 
an 
eigenvalue of $U$ and the
other eigenvalues 
$\lambda_{k,1}$, $\lambda_{k,2}$ are mutually conjugated i.e. 
$\lambda_{k,1}\lambda_{k,2}=1$.
Moreover, eigenvalues for $k'=2\pi-k$ are the same.
\end{lemma}
\begin{proof}
Let us derive the characteristic polynomial for eigenvalues of the 
step operator on the subspace corresponding to the $\ket{\phi_k}$ on the 
position space
\begin{equation}
\overline{U_k}=\diag(e^{-\ii k},e^{\ii k},e^{-\ii ka}) G.
\end{equation}
From assumptions we obtain that $n|la$ and thus $e^{\ii ka}=1$. Thus the 
characteristic 
polynomial 
simplifies to 
\begin{equation}
(1-\lambda)\left(\frac{\lambda}{3}(1-2\cos k)+1+\lambda+\lambda^2\right),
\end{equation}
with real coefficients and the same solution for $-k$ thus $\lambda_0=1$ and 
lemma holds.
The explicit formulas for the eigenvalues with substitution for $e^{\ii 
k}=\omega$  are
$$
\left\{1,-\frac{\omega ^2+4 \omega +(\omega -1) \sqrt{\omega  (\omega
+10)+1}+1}{6 \omega },-\frac{\omega ^2+4 \omega -(\omega -1)
\sqrt{\omega  (\omega +10)+1}+1}{6 \omega }\right\}
$$
and eigenvectors without normalization factors read
\begin{equation}
\left(
\begin{array}{ccc}
 \frac{2}{\omega +1} & \frac{2 \omega }{\omega +1} & 1 \\
 \frac{-\omega +\sqrt{\omega  (\omega +10)+1}-1}{4 \omega } & \frac{1}{4} \left(-\omega -\sqrt{\omega  (\omega +10)+1}-1\right) & 1 \\
 -\frac{\omega +\sqrt{\omega  (\omega +10)+1}+1}{4 \omega } & \frac{1}{4} \left(-\omega +\sqrt{\omega  (\omega +10)+1}-1\right) & 1 \\
\end{array}
\right)
\end{equation}
\end{proof}

\begin{proof}[Proof of Theorem \ref{th:cyclic}]
In order to prove the cyclic property of the walk we consider the 
limiting distribution from Def, \ref{def::limiting} in the alternative 
form
\begin{equation}\label{eq:limiting-formula}
\Pi(x)=\sum_{\lambda}\sum_{(k,j),(k',j')\in V_\lambda ,d}a_{k,j}^{\phantom*} 
a_{k',j'}^{*} \braket{d,x}{\psi_{k,j}}\braket{\psi_{k',j'}}{d,x}.
\end{equation}

We aim at proving that $\Pi(x)=\Pi(x+a)$.
We note that 
for each eigenvalue $\lambda$ such that $|V_{\lambda}|=1$ its 
unique 
eigenvector $\ket{\psi_{k,j}}$ satisfies
$|\braket{d,x}{\psi_{k,j}}|^2=|\braket{d,x+a}{\psi_{k,j}}|^2$
as a result of Lemma 
\ref{lem:eigenvalues}.
We assume that multiple eigenvalues $\lambda_{k,j}=\lambda_{k',j'}$ follow the 
case $k=\frac{2 \pi l}{\mathrm{GCD}(a, n)}, k'=\frac{2 \pi 
(n-l)}{\mathrm{GCD}(a, 
n)}$ resulting from 
construction in Lemma \ref{lem:multi-eigenvalues} and give the 
$e^{\ii ka}=e^{\ii k'a}=1$ and $e^{\ii k}e^{\ii k'}=1$.
In particular, we show that the relative phases that occur in the terms of the sum 
Eq. (\ref{eq:limiting-formula}) vanish every $a$ steps.
Thus the equality
\begin{equation}
\braket{d,x+a}{\psi_{k,j}}\braket{\psi_{k',j'}}{d,x+a} = 
\braket{d,x}{\psi_{k,j}}\braket{\psi_{k',j'}}{d,x}
\end{equation}
holds for $(k,j),(k',j')\in V_\lambda$ and overall limiting probability is 
periodic with the period equal to 
 $\mathrm{GCD}(a, n)$.
\end{proof}

Thus, we observe interesting phenomena that if the paths generated by the lively
steps do not interact \ie\ create separate classes of nodes, then the asymptotic
probabilities become equal within these classes. This behavior is illustrated
in Fig.~\ref{fig:asymptotic-classes}. One should note that the values of the 
limiting distribution depends on the initial state. However, the periodicity is 
not affected by the initial state.

\begin{figure}[ht!]
\centering
    \subfigure[Aperiodic case\label{fig:classes1}]{
    \begin{tikzpicture}

    \def \n {7}
    \def \radius {2.2cm}

    \foreach \s in {1,...,\n}
    {   
    \pgfmathtruncatemacro{\label}{\n-\s};
    \node[draw, circle] (\label) at ({360/\n * (\s+2) }:\radius) {};
    }

    \draw[dashed,thick] (6) to (0);

    \def \n {6}
    \foreach \s in {1,...,\n}
    {
        \pgfmathtruncatemacro{\ss}{\s-1};
        \draw[dashed,thick] (\ss) to (\s);
    }

    \def \n {4}
    \foreach \s in {0,...,\n}
    {   
        \pgfmathtruncatemacro{\ss}{\s+2};
        \draw[dashed,red,thick] (\s) to (\ss);
    }

    \draw[dashed,red,thick] (5) to (0);
    \draw[dashed,red,thick] (6) to (1);
    \end{tikzpicture}
    }
    \subfigure[Periodic case\label{fig:classes2}]{
    \begin{tikzpicture}

    \def \n {6}
    \def \radius {2.2cm}

    \foreach \s in {1,...,\n}
    {   
        \pgfmathtruncatemacro{\label}{\n-\s};
        \node[draw, circle] (\label) at ({360/\n * (\s+2) }:\radius) {};
    }

    \draw[dashed,thick] (5) to (0);

    \def \n {5}
    \foreach \s in {1,...,\n}
    {
        \pgfmathtruncatemacro{\ss}{\s-1};
        \draw[dashed,thick] (\ss) to (\s);
    }
  
    \draw[dotted,thick,red] (0) to (2);
    \draw[dotted,thick,red] (2) to (4);
    \draw[dotted,thick,red] (4) to (0);
    
    \draw[dashdotted,thick,blue] (1) to (3);
    \draw[dashdotted,thick,blue] (3) to (5);
    \draw[dashdotted,thick,blue] (5) to (1);

    \end{tikzpicture}
    }
    
    \caption{Illustration of the simple networks for the lively walk with $n=6$
    and $n=7$ nodes and the liveliness $a=2$. In the case \subref{fig:classes1}
    all nodes belong to the same class. In the case \subref{fig:classes2} red
    (dotted) and blue (dot-dashed) links connect the nodes corresponding to
    different classes exhibiting different asymptotic behavior.}
    \label{fig:asymptotic-classes}
\end{figure}
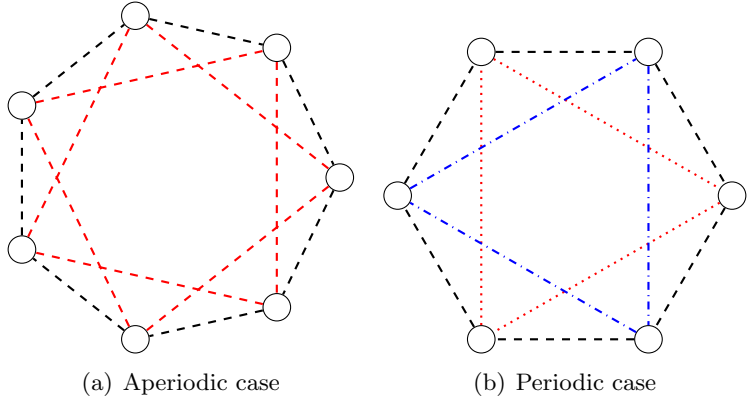

\section{Additional properties of lively walks}\label{sec:liveliness}
In this section we focus on some additional properties of the
introduced model. We start by demonstrating that the model enables us to avoid
the trapping of the particle for any coin. Next we study how the periodicity of
the lively walks is disturbed if one of the links in the network is missing. 

\subsection{Mean difference of position}

Let us consider the situation when we have two parties (or
players) using the introduced model to execute a quantum walk on a network. The
first player aims at exploring the network, whilst the other one aims at
stopping the exploration by trapping the quantum walk. The second party can
choose the coin used during the walk. Below we demonstrate that in such
situation the introduced model is not vulnerable for the actions of the second
party and the trapping can be avoided for any coin.
\def\Expectation{\ensuremath{\mathbb{E}}}

To be more precise we introduce a random variable $(X^\rho, D^\rho)$ that 
models the measurement in canonical basis of the joint position and coin 
register. Using this concept trapping of a quantum walk would mean that the 
probability distribution of the variable $X^{\rho_t}$, for $\rho_t$ being the 
state of the system in time $t$, does not change in time. In particular the 
expectation value $\Expectation(X^{\rho_t})$ would be time-independent.
Thus, in order to avoid trapping we will ensure that the expectation value of 
the position changes in time.
\begin{theorem}
Let $\mathcal{H}_C=\mathbb{C}^3$, $\mathcal{H}_P=\mathbb{C}^n$ be Hilbert 
spaces and $\mathcal{D}(\mathcal{H}_C)$ and $\mathcal{D}(\mathcal{H}_P)$ spaces 
of density operators corresponding to them. Let us suppose that we have a
lively walk with an initial state $\rho_0\in \mathcal{D}(\mathcal{H}_C 
\otimes\mathcal{H}_P)$ in the form  $\rho_0=  
\frac{1}{3}(\ketbra{0}{0}+\ketbra{1}{1}+\ketbra{2}{2})\otimes{\ketbra{x_0}{x_0}}$.
Then for 
arbitrary (time-dependent) three-dimensional coin operator $C_t$ the lively 
walk with liveliness $a\neq 0$,  can not be trapped 
\ie~the difference of the expectation value of the position 
$\Expectation(X^{\rho_{t+1}})-\Expectation(X^{\rho_{t}})$ in time $t$ is 
non-zero.
\end{theorem}

\noindent One should note that in principle $C_t$ can change at each step of the
evolution. For this reason this theorem can be applied for the model with the
time-dependent coin operator.

\medskip

\begin{proof}
	First, we note that the shift operator from Eq.~(\ref{eq:shift}) may be 
	represented in a different form
	\begin{equation}
	S^{(n,a)} = \sum\limits_{c=0}^2\ketbra{c}{c}\otimes U_c^{(n,a)},
	\end{equation}
	where $U_c^{(n,a)}=\sum_x\ketbra{x+\Delta_c^{(a)}}{x}$, 
	$\Delta_c^{(a)}=-\delta_{c,0}+\delta_{c,1}+a\delta_{c,2}$ and we perform 
	addition modulo $n$.
	We also note that each pair of matrices $U_c^{(n,a)}$ 
	commute. Let $\{\ket{e_m}\}_{m=1}^{n}$ be eigenvectors of matrices 
	$U_c^{(n,a)}$.
	By mathematical induction we will show that state after arbitrary number of 
	evolutions is of the form
	
	\begin{alignat}{1}
	\rho = 
	\sum\limits_{ij}\sum\limits_{ml}\rho_{ij,ml}\ketbra{i}{j}\otimes\ketbra{e_m}{e_l},
	\label{eq:ind_assumption}
	\end{alignat}
	where 
	\begin{equation}
	\rho_{ij,mm} =\frac{1}{3}\delta_{ij}a_m \textrm{ with }\ a_m= 
	\braket{e_m}{x_0} 
	\braket{x_0}{e_m}
	\textrm{ and } \ \sum\limits_{m}a_m=1.
	\label{eq:ind_assumption2}
	\end{equation}
	
	Let us show that the relation in 
	Eq.~(\ref{eq:ind_assumption2}) are
	valid for the initial state. Matrix elements of $\rho_0$ are of the form 
	\begin{equation}
	\begin{split}
	\bra{i,e_m}\rho_0\ket{j,e_m}&=\bra{i,e_m} 
	\left(\frac{\1_3}{3}\otimes\ket{x_0} 
	\bra{x_0}\right)\ket{j,e_m}\\
	& =\frac{1}{3}\braket{i}{j}\otimes\braket{e_m}{x_0} 
	\braket{x_0}{e_m}=\frac{1}{3}a_m\delta_{ij},
	\label{exeption_initial}
	\end{split}
	\end{equation}
	where $\ket{i}, \ket{j}$, $i,j = 0,1,2$, are the states of the canonical coin basis.

	Next we will demonstrate inductive step. 
	\begin{equation}
	\begin{split}
	\bra{i,e_m}U\rho U^\dagger\ket{j,e_m}&=\bra{i,e_m}S^{(n,a)}(C_t\otimes\Id_n)\rho
	(C_t^\dagger\otimes\Id_n)S^{(n,a)\dagger}\ket{j,e_m}\\
	&=e^{\ii \alpha_{i,m}}\bra{i,e_m}(C_t\otimes\Id_n)\rho(C_t^\dagger
	\otimes\Id_n) \ket{j,e_m}e^{-\ii\alpha_{j,m}}\\
	&=e^{\ii 
	(\alpha_{i,m}-\alpha_{j,m})}\sum\limits_{c,c'}(\bra{i}C_t\ket{c}\bra{c,e_m})
	\rho(\bra{c'}C_t^\dagger\ket{j}\ket{c',e_m})\\
	&=e^{\ii (\alpha_{i,m}-\alpha_{j,m})}\sum\limits_{c,c'}\bra{i}C_t\ket{c}
	\bra{c'}C_t^\dagger\ket{j}\bra{c,e_m}
	\rho\ket{c',e_m}\\
	&= e^{\ii (\alpha_{i,m}-\alpha_{j,m})}           
	\frac{1}{3}a_m\sum\limits_{c}\bra{i}C_t\ket{c}\bra{c}C_t^\dagger\ket{j}\\
	&= e^{\ii (\alpha_{i,m}-\alpha_{j,m})}           
	\frac{1}{3}a_m\delta_{ij}=\frac{1}{3}a_m\delta_{ij},
	\end{split}
	\end{equation}
	where $e^{\ii \alpha_{i,m}}$ is eigenvalue of $S$ which is corresponding to 
	eigenvector $\ket{i,e_m}$.

	Now we can show that the state from Eq.~(\ref{eq:ind_assumption}) after 
	partial 
	trace on position register is of the form $\frac{\Id_3}{3}$.
	\begin{equation}
	\begin{split}
	\Tr_P(\rho)&=\sum\limits_k 
	\Id_3\otimes\bra{e_k}\Big(\sum\limits_{ij}\sum\limits_{ml}
	\rho_{ij,ml}\ketbra{i}{j}\otimes\ketbra{e_m}{e_l}\Big)\Id_3\otimes\ket{e_k}\\
	&=\sum\limits_k\sum\limits_{ij}\sum\limits_{ml}\rho_{ij,ml}\ketbra{i}{j}
	\braket{e_k}{e_m}\braket{e_l}{e_k}\\
	&=\sum\limits_{ij}\sum\limits_{m}\frac{1}{3}a_m\delta_{ij}\ketbra{i}{j}=\frac{\1_3}{3}.
	\end{split}
	\end{equation}
	
	The above result implies that after an arbitrary number of steps, a quantum 
	walker moves right, left or jumps with equal probabilities.
	
	Let $X^{\rho_{t}},D^{\rho_{t}}$ be random variables corresponding to the 
	measurement on 
	the position and
	coin register after $t$ steps, respectively.
	 Furthermore let $\rho_t^C=(C\otimes \1_n)\rho_t (C\otimes 
	\1_n)^\dagger.$ We have

	\begin{equation}
	\begin{split}
	\Expectation(X^{\rho_{t+1}})&=\frac{1}{3}\Expectation(X^{\rho_{t+1}}| 
	D^{\rho_{t+1}}=0)
	+\frac{1}{3}\Expectation(X^{\rho_{t+1}}|D^{\rho_{t+1}}=1)
		+\frac{1}{3}\Expectation(X^{\rho_{t+1}}|D^{\rho_{t+1}}=2)\\
	&=\frac{1}{3}\Big(\Expectation(X^{\rho_{t}^C}|D^{\rho_{t}^C}=0)-1\Big)
	+\frac{1}{3}\Big(\Expectation(X^{\rho_{t}^C}|D^{\rho_{t}^C}=1)+1\Big)
	+\frac{1}{3}\Big(\Expectation(X^{\rho_{t}^C}|D^{\rho_{t}^C}=2)+a\Big)\\
	&=\frac{1}{3}\Big(\Expectation(X^{\rho_{t}^C}|D^{\rho_{t}}=0)+\Expectation(X^{\rho_{t}^C}|D^{\rho_{t}}=1)+
	\Expectation(X^{\rho_{t}^C}|D^{\rho_{t}}=2)\Big)+\frac{1}{3}\Big(-1+1+a\Big)\\
	&=\Expectation(X^{\rho_{t}^C})+\frac{1}{3}(-1+1+a)\\
	&=\Expectation(X^{\rho_{t}})+\frac{a}{3},
	\end{split}
	\end{equation} 
	where $P(D^{\rho_{t+1}}=i)=\frac{1}{3}$ for all $i\in\{0,1,2\}$.
	This proves the claim.    
\end{proof}

One should note that this enables the user to obtain an arbitrary change of
position. This does not depend on whether the $\mathrm{GCD}(a,n)>1$.

\subsection{Networks with broken links}
As the second application of the introduced model we describe a simple method of
detecting link failures in the network~\cite{feldman10structural}.
 
Let us consider the network delivering the connection for the implementation of
the lively quantum walk. In such situation the limiting distribution will have
the properties described in Section~\ref{sec:periodicity}.

Let us now consider a failure of the network, which can be described as a lack
of one of the links (see~Fig.~\ref{fig:broken-link}). In this case the
\emph{broken link} can be understood as a connection error or as an action of a
malicious party.

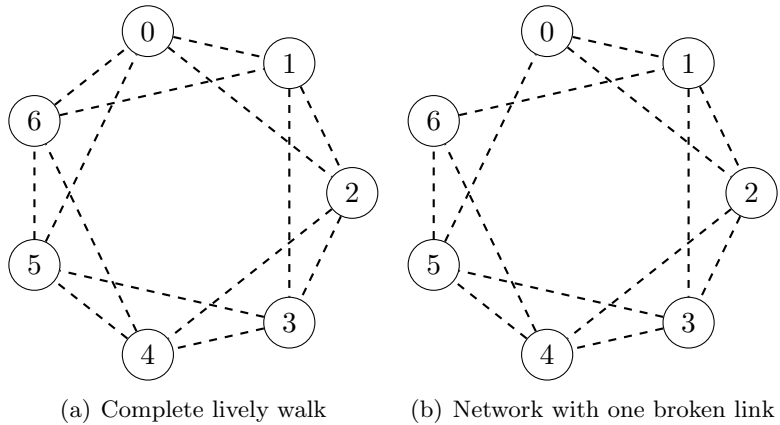
\begin{figure}[ht!]
\centering
\subfigure[Complete lively walk\label{fig:cycle-full}]{
\begin{tikzpicture}

\def \n {7}
\def \radius {2.2cm}

\foreach \s in {1,...,\n}
{   
    \pgfmathtruncatemacro{\label}{\n-\s};
    \node[draw, circle] (\label) at ({360/\n * (\s+2) }:\radius) {\label};
}

\draw[dashed,thick] (6) to (0);

\def \n {6}
\foreach \s in {1,...,\n}
{
    \pgfmathtruncatemacro{\ss}{\s-1};
    \draw[dashed,thick] (\ss) to (\s);
}

\def \n {4}
\foreach \s in {0,...,\n}
{   
    \pgfmathtruncatemacro{\ss}{\s+2};
    \draw[dashed,thick] (\s) to (\ss);
}

\draw[dashed,thick] (5) to (0);
\draw[dashed,thick] (6) to (1);
\end{tikzpicture}
}
\subfigure[Network with one broken link\label{fig:cycle-broken}]{
\begin{tikzpicture}

\def \n {7}
\def \radius {2.2cm}

\foreach \s in {1,...,\n}
{   
    \pgfmathtruncatemacro{\label}{\n-\s};
    \node[draw, circle] (\label) at ({360/\n * (\s+2) }:\radius) {\label};
}
\def \n {6}
\foreach \s in {1,...,\n}
{
    \pgfmathtruncatemacro{\ss}{\s-1};
    \draw[dashed,thick] (\ss) to (\s);
}

\def \n {4}
\foreach \s in {0,...,\n}
{   
    \pgfmathtruncatemacro{\ss}{\s+2};
    \draw[dashed,thick] (\s) to (\ss);
}
\draw[dashed,thick] (5) to (0);
\draw[dashed,thick] (6) to (1);
\end{tikzpicture}
}

    \caption{Illustration of the simple network with one broken link for the
    lively walk with $n=7$ nodes and $a=2$. For the standard quantum walker the
    broken cycle (case~\subref{fig:cycle-broken}) is equivalent to the line
    segment.}
    \label{fig:broken-link}
\end{figure}

We can define the model used to describe a lively walk on a cycle with one broken
link. We assume that the walker is able to execute moves with $a=2$. In such
case the shift operator is given as
\begin{equation}
\begin{split}
S_B^{n,2} = &\sum_{x=1}^{n-2} S_{x}^{n,2}
+ \ketbra{1}{0}\otimes\ketbra{0}{0}
+ \ketbra{1}{1}\otimes\ketbra{1}{0}
+ \ketbra{2}{2}\otimes\ketbra{2}{0} \\
&+ \ketbra{0}{1}\otimes\ketbra{n-1}{n-1}
+ \ketbra{0}{0}\otimes\ketbra{n-2}{n-1}
+ \ketbra{2}{2}\otimes\ketbra{0}{n-2}
\end{split}.
\label{eqn:lively-broken}
\end{equation}

\begin{figure}[ht!]
    \centering
    
\subfigure[$a=5$\label{fig:fullbroken5}]{\includegraphics{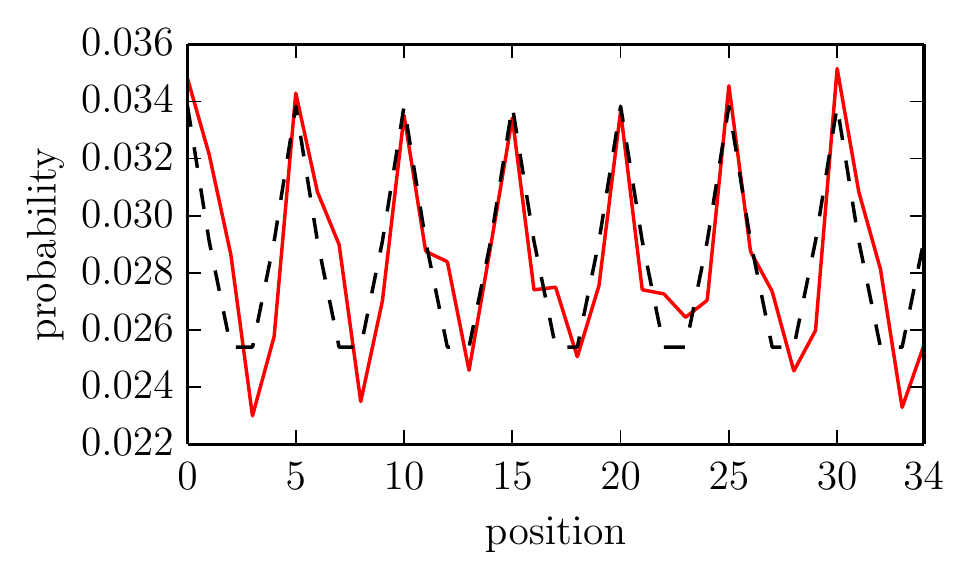}}
    
\subfigure[$a=7$\label{fig:fullbroken7}]{\includegraphics{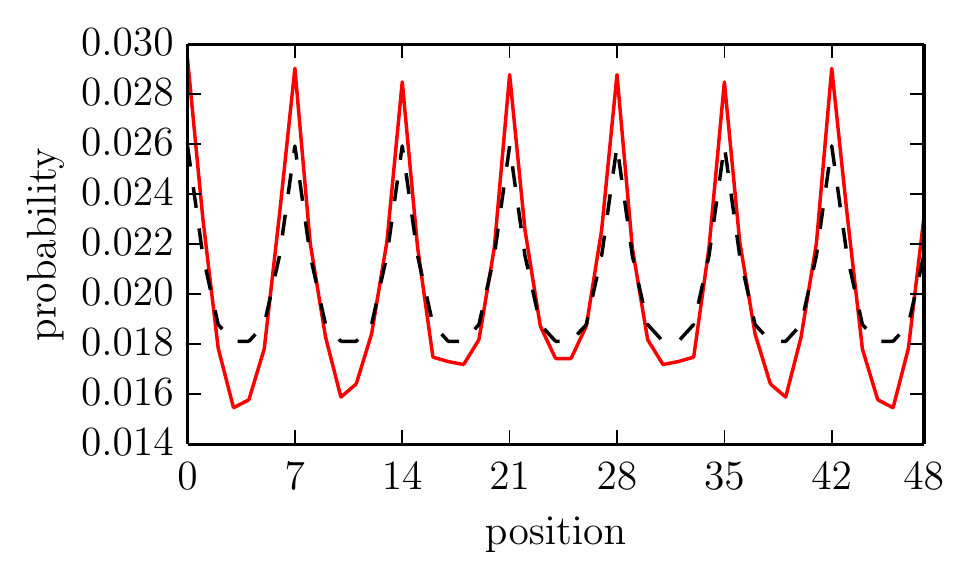}}\\
    
\subfigure[$a=11$\label{fig:fullbroken11}]{\includegraphics{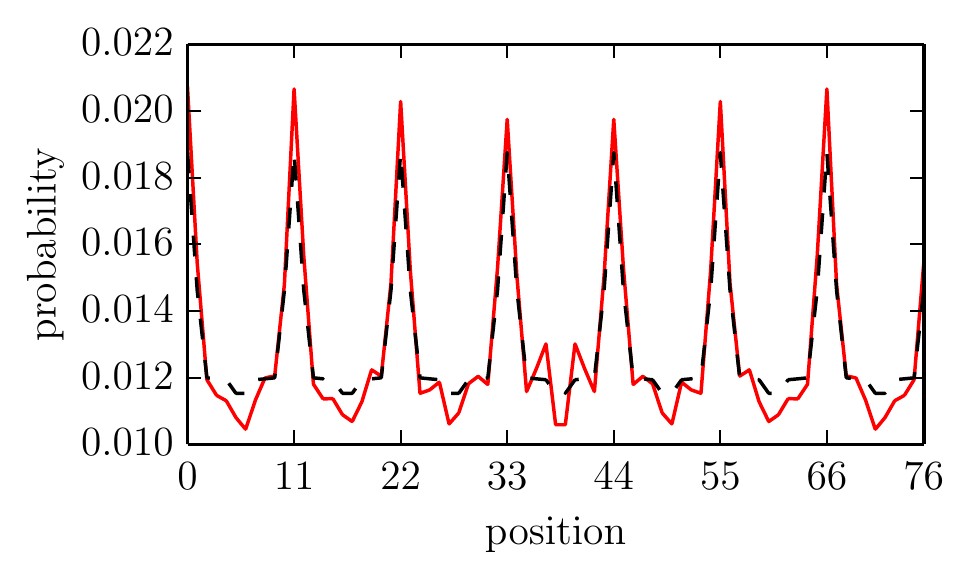}}
    \caption{Influence of one broken link on the time-averaged limiting
    probability distribution. Limiting distributions for the cases without
    broken link are plotted in black dotted line, whlist limiting 
distributions 
    for the cases with one broken link are plotted in solid red line.
    The periodicity of the time-averaged
    limiting distribution is disturbed around the nodes without the connection.
    Here the broken link is located between nodes 18 and 19 (case
    \subref{fig:fullbroken5}), 24 and 25 (case \subref{fig:fullbroken7}), and 
38 and 39
    (case \subref{fig:fullbroken11}), respectively and the initial state is of 
the 
    form $\frac{1}{\sqrt{3}}(\ket{0}+\ket{1}+\ket{2})\otimes\ket{0}$.} 
    \label{fig:broken-limiting}
\end{figure}

The second part of Eq.~(\ref{eqn:lively-broken}) can be interpreted as the
laziness condition -- \eg\ in position $\ket{0}$ with coin
pointing in the direction of a broken link, the walker does not move and the
coin is changed from $\ket{0}$ to $\ket{1}$. However, one should note that due to the
possibility of executing steps with a larger range, the cycle with broken links
is not equivalent to a line segment.

The time-averaged limiting distribution in this situation is presented in
Fig.~\ref{fig:broken-limiting}. One can easily observe that the situation where
one of the links is missing has a significant
impact on the periodicity of the limiting distribution. The
disturbance is particularly strong around the location of the broken link.



\begin{figure}[ht!]
    \centering
	\subfigure[$a=5$\label{fig:chernoff5}]{\includegraphics{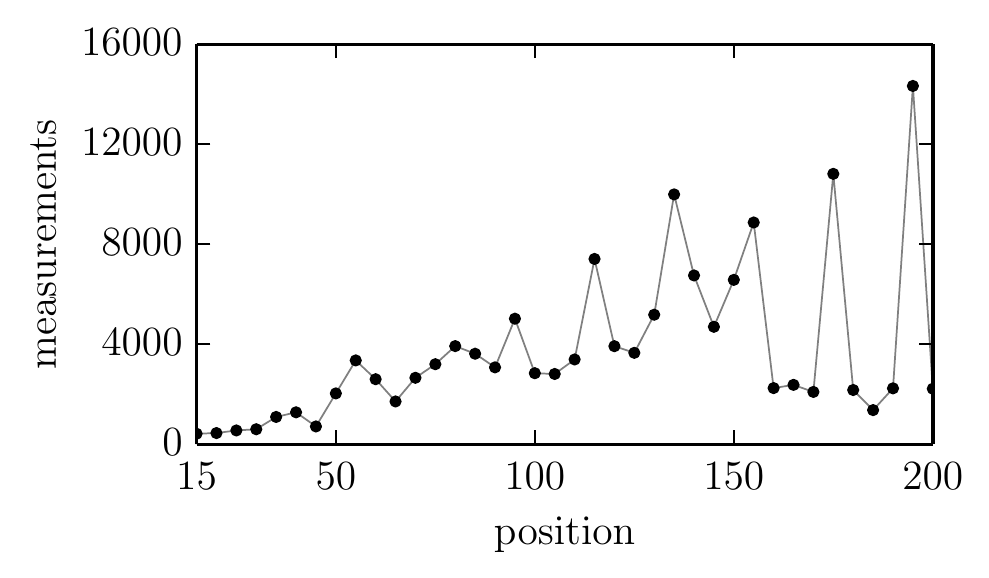}}
	\subfigure[$a=7$\label{fig:chernoff7}]{\includegraphics{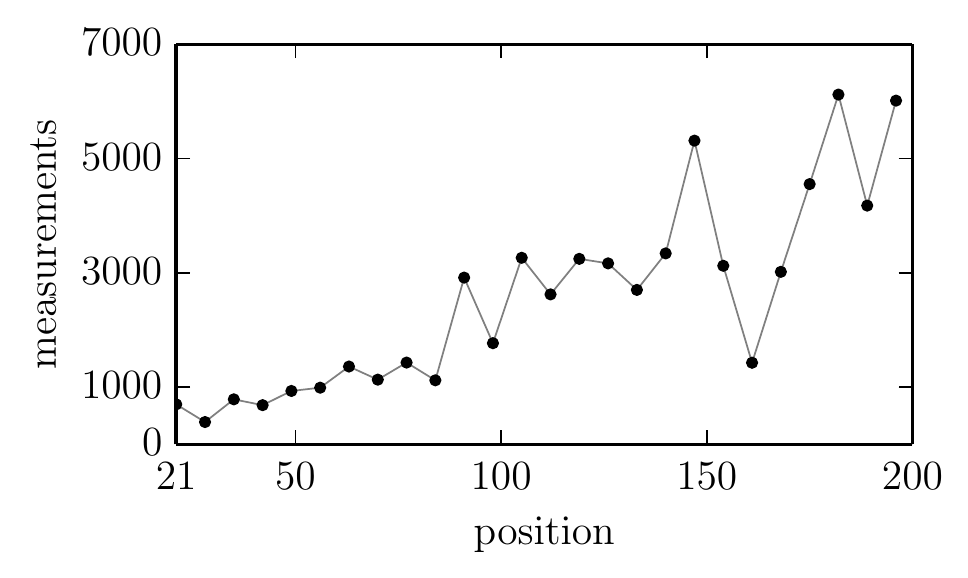}}\\
	\subfigure[$a=11$\label{fig:chernoff11}]{\includegraphics{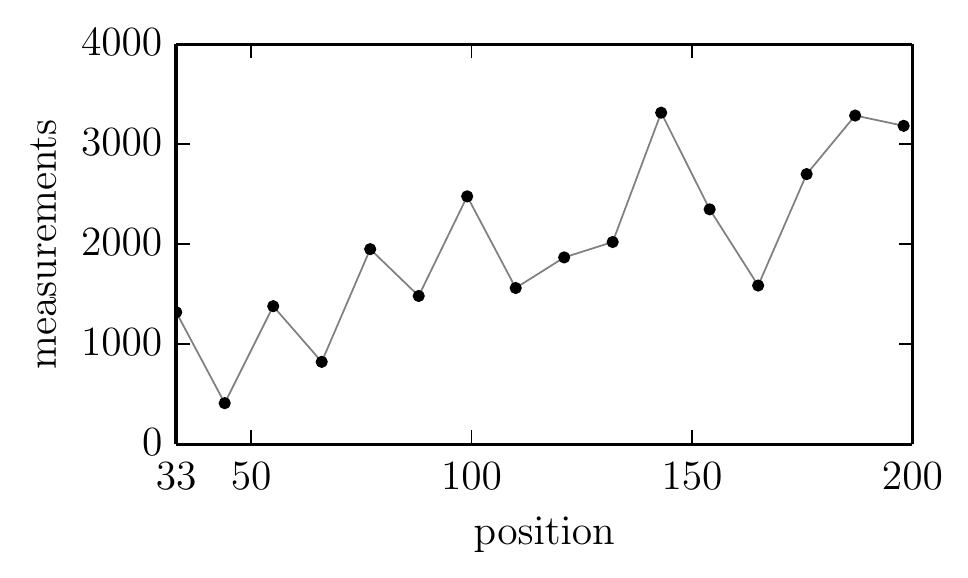}}
    \caption{Number of measurement required to assess if the
    network has one broken link. The plots represent results obtained for
    networks with number of nodes being a multiplicity of the liveliness
    parameter. Broken links are always between $\lfloor \frac{n}{2}\rfloor$ and 
    $\lfloor \frac{n}{2}\rfloor+1$ nodes, where $n$ is number of nodes and the 
    initial state is of the 
    form$\frac{1}{\sqrt{3}}(\ket{0}+\ket{1}+\ket{2})\otimes\ket{0}$.}
    \label{fig:chernoff}
\end{figure}

Suppose we know the limiting distribution for the networks
without and with a broken link, and denote them respectively $P_0$ and $P_1$. We
want to asses if the given network has a broken link. In other words we need to
make a statistical test which tells us if the considered model has distribution
$P_0$ or $P_1$. Our goal is to estimate how many measurements are needed to
discover one broken link and how this number depends on the size of the
network.

This problem is equivalent to finding the relationship between
the best ability to distinguish between statistical hypothesis ($H_0:\ P_0$ and
$H_1:\ P_1$) and the size of a statistical sample (the number of measurements)
\cite{baigneres2008complexity}. The ability of distinguishing hypothesis is
called the advantage and is expressed as
\begin{equation}
Adv_{q}(P_0,P_1)=|1-\alpha - \beta|,
\end{equation}
where $q$ denotes a number of statistical samples, and $\alpha$
and $\beta$ are errors of type I and II, respectively. 
It turns out that there is a connection between the number of samples and the 
best advantage for
distinguishing between distributions $P_0$ and $P_1$
\begin{equation}
\label{bestAdvantage}
1-BestAdv_q(P_0,P_1)\doteq2^{-qC(P_0,P_1)},
\end{equation}
where $\doteq$ denotes that two functions are asymptotically
equivalent and $C(P_0,P_1)$ is called Chernoff information
\begin{equation}
\label{eq:Chernoff_inf}
C(P_0,P_1)=-\inf\limits_{0<\lambda<1}\log\sum\limits_i P_0(i)^\lambda 
P_1(i)^{1-\lambda}.
\end{equation}
From Eq. (\ref{bestAdvantage}) we obtain that if we have
$q\approx\frac{1}{C(P_0,P_1)}$ samples, then we have the maximal ability to
distinguish $P_0$ from $P_1$.

The number of measurements required in the case of lively quantum
walk with one broken link is presented in Fig.~\ref{fig:chernoff}. It is 
important to say, that results presented in Fig.~\ref{fig:chernoff} are 
obtained by numerical approximation of Chernoff 
information~(\ref{eq:Chernoff_inf}). One can
observe that for large networks this number decreases with the growing
liveliness parameter. This reflects the fact that the smaller networks and for
the small liveliness parameter, the density of connections is higher. In such
case one broken link does not have big impact on the network structure and, as
the result, on the time-averaged limiting distribution. On the other hand, for
larger networks and large liveliness parameter, the density of connections
becomes smaller. This results in more significant impact of the broken
connection on the observed behavior of the walk.

\section{Concluding remarks}\label{sec:concluding}
In this work we have introduced and studied a parametrized model
of a quantum walk on cycle. The introduced model can be used to
study the situation where the near-neighbor communication in the network is
supplemented by the existence of the long-range links between the selected
nodes.

The introduced model displays the periodicity of the
time-averaged limiting distribution. We have proved that the periodicity of the limiting
distribution is connected with the liveliness parameter.

The existence of additional connections enables the
utilization of the introduced model for the purpose of quantum network
exploration.
In particular, the additional connection allows avoiding the
trapping of the walker for any choice of the coin. This makes the lively walk
resistant to the actions of a malicious party disturbing the programme of the
quantum walker exploring the network.

Thanks to the introduction of additional connections, the
lively walk can preserve its properties in the situation when one of the links
is missing. This represents the situation when a structural error in the network
occurs. We have shown that such errors disturb the periodicity properties of the
introduced model. Moreover, one can argue that the additional connections can be
beneficial from the point of view of quantum walk integrity. This is due to the
fact that the additional connections make the network more resistant to broken
links.

\paragraph*{Acknowledgements}
Work of PS has been supported by Polish National Science Centre under the
research project 2013/11/N/ST6/03030, JAM has been supported by Ministry of
Science and Higher Education under Iuventus Plus project IP 2014 031073 and MO
has been supported by Polish National Science Centre under the research project
2011/03/D/ST6/00413.

Authors would like to thank {\L.}~Pawela and Z.~Pucha{\l}a for helpful
suggestions and discussions during the preparation of the manuscript.

\bibliographystyle{unsrt}
\bibliography{lively_qwalks_cycles}

\end{document}